\documentclass[11pt]{article}
\usepackage{amssymb}
\usepackage{amsmath}
\usepackage{txfonts}
\usepackage{type1cm}
\newcommand{\calM}{\mathcal{M}}
\newcommand{\bfR}{\mathbf{R}}
\newtheorem{theorem}{Theorem}
\newtheorem{lemma}[theorem]{Lemma}
\newtheorem{corollary}[theorem]{Corollary}
\newtheorem{proposition}[theorem]{Proposition}
\newtheorem{problem}{Problem}

\newenvironment{proof}{%
  \noindent{\it Proof\ }}{%
  \hspace*{\fill}\qed
  \vspace{2ex}\\}
\newenvironment{proofsketch}{%
  \noindent{\it Proof (Sketch).\ }}{%
  \hspace*{\fill}\qed
  \vspace{2ex}}
\newcommand{\qed}{$\square$}
\newcommand{\bra}[1]{\langle #1 |}
\newcommand{\ket}[1]{| #1 \rangle}
\newcommand{\abs}[1]{\vert #1 \vert}
\newcommand{\ceil}[1]{\lceil #1 \rceil}
\newcommand{\fig}{Fig.}

\newcommand{\pqsubset}{\mbox{\boldmath $(p,q)$}{\bf -subset}}
\newcommand{\kclaw}{\text{\boldmath $k$}\text{\bf -claw}}
\newcommand{\claw}{\text{\bf claw}}
\begin{document}
\sloppy
\title{Claw Finding Algorithms Using Quantum Walk}
\author{Seiichiro Tani\thanks{tani@theory.brl.ntt.co.jp}\\
Quantum Computation and Information Project, ERATO-SORST, JST.\\
NTT Communication Science Laboratories, NTT Corporation.
}
\date{}
\maketitle              

\begin{abstract}
     The claw finding problem has been studied
in terms of query complexity
   as one of the problems closely
     connected to 
   cryptography.
   For given two functions, $f$ and $g$, as an oracle
   which have domains of
    size $N$ and $M$ $(N\leq M)$, respectively, and the same range,
   the goal of the problem 
   is to find $x$ and $y$ such that $f(x)=g(y)$.
    This problem has been considered in both quantum and classical
      settings
    in terms of query complexity.
   This paper describes an optimal algorithm using quantum walk 
that solves this problem.
Our algorithm can be slightly modified to 
solve a more general problem
of finding a tuple consisting of elements
in the two function domains 
that has prespecified property.
Our algorithm can also be generalized to find a
claw
   of $k$ functions for any constant integer $k>1$, where the domains of the
    functions may have different size. 
\end{abstract}

\section{Introduction}
 The most significant discovery in quantum computation
 would be Shor's polynomial-time quantum algorithms
 for factoring integers and computing discrete
 logarithms~\cite{Sho97SIComp}, both of which are believed to 
 be hard to solve in classical settings and are thus used in
 arguments for the security of
the widely used cryptosystems.
 Another significant discovery is Grover's quantum algorithm
 for the problem of searching an unstructured set~\cite{grover96},
i.e,
the problem of searching for $i\in \{0,1,\ldots ,N-1 \}$ such that $f(i)=1$
for a hidden Boolean function $f$; 
it has yielded a variety of 
generalizations~\cite{Bra00AMS,hoyer-mosca-dewolf03,ambainis07SICOMP,szegedy04,magniez-nayak-roland-santha07}.
Grover's algorithm and its generalizations assume the \emph{oracle computation model}, 
in which a problem instance is given as a black box
(called an oracle) and
any algorithm needs to make queries to the
black box to get sufficient information on the instance.
In the case of searching an unstructured set,
any algorithm needs to make queries of the form ``what is the value
 of function $f$ for input $i$ ?'' to the given oracle.
In the oracle computation model, 
the efficiency of an algorithm is usually measured by the number of queries
the algorithm needs to make, i.e., the query complexity of the
 algorithm. The query complexity of a problem means
the query complexity of the algorithm that solves the problem with
fewest queries.

One of the earliest applications of Grover's algorithm
was the bounded-error algorithm of Brassard, H{\o}yer and
Tapp~\cite{brassard-hoyer-tapp98};
it addressed the \emph{collision} problem
in a cryptographic context,
i.e., 
finding pair $(x,y)$ such that $f(x)=f(y)$,
in a given 2-to-1 function $f$ of domain size $N$.
Their quantum algorithm requires $O(N^{1/3})$ queries,
whereas any bounded-error classical algorithm needs $\Theta(N^{1/2})$ queries.
Subsequently, Aaronson and Shi~\cite{aaronson-shi04}
proved the matching lower bound.
Brassard et al.~\cite{brassard-hoyer-tapp98}
considered two more related problems: the \emph{element distinctness}
problem and the \emph{claw finding} problem.
These problems are also important
in a cryptographic sense.
Furthermore, studying these problems
has deepened our understanding of the power of 
quantum computation.

The element distinctness problem is
to decide whether or not $N$ integers given as an oracle are all distinct.
Buhrman et al.~\cite{buhrman-durr-heilgman-hoyer-magniez-santha-dewolf05} 
gave a bounded-error
algorithm for the problem, which makes $O(N^{3/4})$ queries
(strictly speaking, they assumed a comparison oracle,
which returns just the result of comparing function values for two
specified inputs,
and, in this case, 
the query complexity is $O(N^{3/4}\log N)$).
Subsequently, Ambainis~\cite{ambainis07SICOMP} gave an improved upper bound
$O(N^{2/3})$ by introducing a new framework of quantum walk
(his quantum walk algorithm was reviewed from a slightly more general
point of view in~\cite{magniez-santha-szegedy05,childs-eisenberg05},
and a much more general framework was given by Szegedy~\cite{szegedy04}).
This upper bound matches the lower bound proved by Aaronson and
Shi~\cite{aaronson-shi04}.

The \emph{claw
finding problem} is defined as follows.
Given two functions $f:X\rightarrow Z$ and $g:Y\rightarrow Z$ as an
oracle, 
decide whether or not there exists
at least one
pair $(x,y)\in X\times Y$, called a \emph{claw}, such that
$f(x)=g(y)$,
and \emph{find} a claw 
if it exists, 
where $X$ and $Y$ are domains of size $N$ and $M$ $(N\leq M)$,
respectively. By $\claw_\text{finding}(N,M)$, we mean this problem.

After Brassard et al.~\cite{brassard-hoyer-tapp98}
considered a special case of the claw finding problem,
Buhrman et
al.~\cite{buhrman-durr-heilgman-hoyer-magniez-santha-dewolf01} gave a
quantum algorithm
that requires
$O(N^{1/2}M^{1/4})$ queries for $N\leq M< N^2$
and $O(M^{1/2})$ queries for $M\geq N^2$ (strictly speaking, they assumed a comparison oracle,
and, in this case, the query complexity is multiplied by $\log N$).
They also proved that any algorithm requires
$\Omega(M^{1/2})$ queries
by reducing the search problem over an unstructured
set to the claw finding problem.
Thus, while their bounds of the query complexity are
tight when $M\geq N^2$, there is still a big gap when $N\leq M<N^2$.
Furthermore, they considered
the case of $k$ functions, i.e.,
the \emph{$k$-claw finding} problem defined as follows:
given $k$ functions $f_i:X_i:=\{1,\dots ,N_i \}\rightarrow Z\ (i\in \{1,\dots ,k\})$ 
as an oracle, where 
$k>1$ is any constant integer,
and $N_i\leq N_j$ if $i<j$,
decide whether or not there exists at least
one $k$-claw, i.e., a tuple
$(x_1,\dots , x_k)\in X_1\times \dots \times X_k$ such that
$f_i(x_i)=f_j(x_j)$ for any $i,j\in \{1,\dots ,k\}$, 
and find a $k$-claw if it exists.
A generalization of their algorithm works well
for the $k$-claw finding problem;
its query complexity is
$O(N^{1-1/2^k})$ if $N_i=N$ for all $i\in \{1,\dots ,k\}$.
It is shown in~\cite{magniez-santha-szegedy05} 
that the quantum-walk algorithm in~\cite{ambainis07SICOMP} for the element
distinctness problem 
is general enought to be applied with
slight modification to the $k$-claw finding problem;
this yields query complexity 
 $O((\sum _{i=1}^kN_i)^{\frac{k}{k+1}})$ 
 if the promise is assumed that there is at most one solution,
 and, with random reduction,
query complexity $\tilde{O}((\sum _{i=1}^kN_i)^{\frac{k}{k+1}})$ 
for the problem without the single-solution promise.
Zhang~\cite{zhang05COCOON} generalized the quantum-walk algorithm in~\cite{ambainis07SICOMP}
to solve the claw finding problem with the single-solution promise
by making
$O((NM)^{1/3})$ queries for $N\leq M<N^2$ and 
$O(M^{1/2})$ for $M\geq N^2$.
This upper bound is optimal, since the matching lower bound 
$\Omega((NM)^{1/3})$ was proved in the paper 
by reducing the collision problem to the claw finding problem.
Zhang also showed that the algorithm can be generalized to solve
a more general problem of finding a tuple consisting of elements
in the domains of given $k$ functions with the single-solution promise.
To solve the problems without the promise,
we usually use a randomly reduction to the problem \emph{with} the
single-solution promise, which is known to increase
the query complexity by at most a log factor
as pointed out in~\cite{magniez-santha-szegedy05}
(if the problem has certain robust properties, there is a random reduction 
that increases the query complexity by a constant multiplicative factor, e.g., \cite{ambainis07SICOMP}.)

This paper gives an optimal quantum algorithm 
that directly (i.e., without using such a random reduction) solves the claw finding problem without the
single-solution promise.
The query complexity of our algorithm is as follows:
\[
   Q(\claw_\text{finding}(N,M))=
     \begin{cases}
   O\left((NM)^{1/3}\right) & (N\leq M<N^2)\\
   O\left(M^{1/2}\right) & (M\geq N^2),
     \end{cases}
\]
where $Q(P)$ means the number of queries required to solve problem
$P$ with one-sided bounded error
(i.e., with the one-sided error probability bounded by a certain constant, say, 1/3).
The optimality is guaranteed by the lower bounds given in \cite{buhrman-durr-heilgman-hoyer-magniez-santha-dewolf01,zhang05COCOON}.
Our algorithm can be modified to solve a more general
problem
of finding
a tuple $(x_1,\dots ,x_p,y_1,\dots ,y_q)\in X^p\times Y^q$ 
such that
${x_i\neq x_j}$ and $y_i\neq y_j$ for any $i\neq j$, and 
$(f(x_1),\dots ,f(x_p),g(y_1),\dots ,g(y_q))\in R$,
for given $R\subseteq Z^{p+q}$, where $p$ and $q$ are positive
constant integers. We call this problem \emph{$(p,q)$-subset
  finding problem}
and denote it by $\pqsubset _{\rm finding}(N,M))$.
Thus, $\claw_\text{finding}(N,M)$ is a special case of 
$\pqsubset_\text{finding}(N,M))$ with $p=q=1$ and equality relation $R$.
The query complexity is
\[
Q(\pqsubset_\text{finding}(N,M))=
\begin{cases}
O((N^pM^q)^{1/(p+q+1)}) & {N\leq M<N^{1+1/q}}\\
O(M^{q/(1+q)}) & {M\geq N^{1+1/q}}.
\end{cases}
\]
Our claw finding algorithm first finds subsets $\tilde{X}\subseteq X$ and $\tilde{Y}\subseteq
Y$ of size $O(1)$ such that there is a claw in $\tilde{X}\times \tilde{Y}$, 
by
using binary and 4-ary searches over $X$ and $Y$; in order to
decide which branch we should proceed at each visited node in the
search trees, we use a subroutine that \emph{decides},
with one-sided bounded error,
whether or not there 
exists a claw of two functions $f$ and $g$.
The algorithm then searches $\tilde{X}\times \tilde{Y}$ for a claw
by making classical queries.
If we na\"{\i}vely repeated the bounded-error subroutine 
$O(\log M)$ times at each visited node
to guarantee bounded error as a whole,
a ``log'' factor would be multiplied to 
the total query complexity.
Instead, 
at the node of depth $s$ in the search trees,
we repeat the subroutine $O(s)$ times
to amplify success probability. This achieves bounded error as a
 whole,
while pushing up the query complexity by just a constant multiplicative
 factor.
This binary search technique can be used to solve other problems such
as the search version of 
the element distinctness problem, with the quantum walk algorithm for
the problems in \cite{szegedy04}.
(H{\o}yer et al.~\cite{hoyer-mosca-dewolf03}
introduced an error reduction technique with a similar flavor; however,
their technique is used in an algorithmic context different from ours:
their error reduction is performed at each recursion level
while ours is sequentially used at each step of the search tree.)

The subroutine is developed around the Szegedy's quantum walk
framework~\cite{szegedy04} over a Markov chain on the graph
categorical product of two Johnson graphs, which
correspond to the two functions (with an idea similar to the one used
in~\cite{buhrman-spalek06}). 
The \emph{Johnson graph} $J(n,k)$ is a connected regular graph with 
${n \choose k}$ vertices such that every vertex is a subset of size
$k$ of $[n]$; two vertices are adjacent if and only if 
the symmetric difference of their corresponding subsets has size 2.
For two functions $f$ and $g$ with domains $X$ and $Y$ such that
$|X|\leq |Y|$,
the subroutine applies Szegedy's quantum walk
to the graph categorical
product of two Johnson graphs $J_f=J(|X|,(|X||Y|)^{1/3})$ and $J_g=J(|Y|,(|X||Y|)^{1/3})$
if $|Y|\leq |X|^2$, and $J_f=J(|X|,|X|)$ and $J_g=J(|Y|,|X|)$ otherwise.

Our algorithm can be generalized to the $k$-claw
finding problem.
For the $k$-claw finding problem ${\kclaw _\text{finding}(N_1,\dots,N_k)}$ against the $k$ functions
with domain sizes $N_i\ (i=1,\dots, k)$, repectively,
\[
Q(\kclaw_\text{finding}(N_1,\dots,
N_k))=
  \begin{cases}
O\left(\left(\prod_{i=1}^k N_i\right)^\frac{1}{k+1}\right) & \mbox{if $\prod_{i=2}^{k}N_i=O(N_1^k)$},\\
O\left(\sqrt{\prod_{i=2}^k N_i/N_1^{k-2}}\right)
&\mbox{otherwise.}
  \end{cases}
\]

Our algorithms can work with slight modification
even 
against a comparison oracle 
(i.e., against an oracle that,
for a given pair of inputs $(x_i,x_j)\in X_i\times X_j$, only
decides which is the larger of two function values $f_i(x_i)$ and $f_j(x_j)$);
the query
complexity increases by a multiplicative factor of $\log N_1$ for
the $k$-function case ($\log N$ for the two-function case).

\subsection*{Related works}
Recently, Magniez et al.~\cite{magniez-nayak-roland-santha07}
developed a new quantum walk over a Markov chain.
One of the advantages of their quantum walk over Szegedy's quantum
walk is that their quantum walk can 
\emph{find} a marked vertex if there is at least one marked vertex,
which would simplify our algorithm.
Interestingly, 
our algorithm shows 
Szegedy's quantum walk 
together with carefully adjusted 
binary search
can find a solution in some interesting 
problems such as the claw finding problem and the element distinctness
problem with the same order of query complexity.

\section{Preliminaries}
This section defines problems and introduces some useful techniques.
We denote the set of positive integers by $\Bbb{Z}^*$,
the set of $\{ i\mid j\leq i \leq k \mbox{ for } i,j,k\in \Bbb{Z}^*
\}$ by $[j.k]$, and $[1.k]$ by $[k]$ for short.
\begin{problem}[Claw Finding Problem]
Given two functions $f:X:=[N]\rightarrow Z$ and $g:Y:=[M]\rightarrow Z$ as an oracle
for $N\leq M$, where
$Z=[\abs{Z}]$, 
find a
pair $(x,y)\in X\times Y$ such that $f(x)=g(y)$ 
if such a pair exists.
\end{problem}
Actually, $Z$ is allowed to be any totally ordered set,
but we adopt the above definition for simplicity.

In a quantum setting, the two functions are given as quantum oracle $O_{f,g}$ which is
defined as 
$
O_{f,g}: \ket{p,z,w}\longrightarrow \ket{p,z\oplus P(p)\pmod {|Z|},w},
$
where $p\in X\cup Y$, $z\in Z$, $w$ is work space, $P(p)$ is defined as $f(p)$ if $p\in
X$
and $g(p)$ if $p\in Y$
(note that it easy to know whether $p$ is in $X$ or $Y$ by
using one more bit to represent $p$).
This kind of oracle, which returns the 
value of the function(s), is called a \emph{standard oracle}.

Another type of oracle is called the \emph{comparison oracle}, which,
for given two inputs,
only decides which is the larger of the two function values
corresponding to the inputs.
More formally, comparison oracle $O_{f,g}$ is defined as
$
O_{f,g}: \ket{p,q,b,w}\longrightarrow \ket{p,q,b\oplus [P(p)\leq Q(q)],w},
$
where $p,q\in X\cup Y$, $b\in \{0,1\}$, $w$ and $P$ are defined as in the
standard oracle, $Q$ is defined in the same way as $P$,
and $[P(p)\leq Q(q)]$ is the predicate
such that its value is 1 if and only if $P(p)\leq Q(q)$.

It is obvious that, if we are given a standard oracle, we can realize
a comparison oracle by issuing $O(1)$ queries to the standard oracle.
Thus, upper bounds for a comparison oracle are those for a standard
oracle, and lower bounds for a standard oracle are those for a
comparison oracle, if we ignore constant multiplicative factors.

Buhrman et
al.~\cite{buhrman-durr-heilgman-hoyer-magniez-santha-dewolf01}
generalized the claw finding problem to a $k$-function case.
\begin{problem}[$k$-Claw Finding Problem]
Given $k$ functions $f_i:X_i:=[N_i]\rightarrow Z\ (i\in [k])$ 
as an oracle, where $N_i\leq N_j$ if $i<j$,
and $Z:=[|Z|]$,
find a $k$-claw, i.e., a $k$-tuple
$(x_1,\dots , x_k)\in X_1\times \dots \times X_k$ such that
$f_i(x_i)=f_j(x_j)$ for any $i,j\in [k]$, 
if it exists.
\end{problem}
Standard and comparison oracles are defined in almost the same way as
in the two-function case, except that
inputs $p$ and $q$ belong to one of $X_i$'s, respectively, for $i\in[k]$.

The next theorem describes Szegedy's framework, which we use 
to prove our upper bounds.
\begin{theorem}[\cite{szegedy04}]
\label{th:szegedy04}
  Let $\calM$ be a symmetric Markov chain with state set $V$ and transition
  matrix $P$ and let
  $\delta_{\calM}$ be the spectral gap of $P$,
i.e., $1-\max _i\abs{\lambda _i}$ for the eigenvalues $\lambda _i$'s of $P$. For a certain subset
  $V'\subseteq V$ with the promise that $|V'|$ is either 0
  or at least $\epsilon |V|$ for $0<\epsilon <1$,
any element in $V'$ is marked. 
For $T=O (1/\sqrt{\epsilon\delta_{\calM} })$, the next quantum
  algorithm decides whether $|V'|$ is 0 (``false'') or at least 
$\epsilon |V|$ (``true'') with one-sided bounded error with cost
$O(C_U+(C_F+C_W)/\sqrt{\delta_{\calM} \epsilon})$,
where $C=\sum _{i}\ket{c_i}\bra{c_i}\ $ for $\ket{c_i}=\sum
  _{j}\sqrt{P_{i,j}}\ket{i}\ket{j}$ and 
$R=\sum _{j}\ket{r_j}\bra{r_j}\ $ for $\ket{r_j}=\sum _{i}\sqrt{P_{j,i}}\ket{i}\ket{j}$:
  \begin{enumerate}
\setlength{\itemsep}{0mm}
\setlength{\parskip}{0mm}
\item Prepare $\ket{0}$ in a one-qubit register $\bfR_0$, and prepare
a uniform superposition 
$
\ket{\phi _0}:=\frac{1}{\sqrt{r|V|}}\sum_{i,j\in V,P_{i,j}\neq 0} \ket{i}\ket{j}
$ in a register $\bfR_1$
with cost at most $C_U$, where $r$ is the number of adjacent states
(of any state) in $\calM$.
\item Apply the Hadamard operator to $\bfR_0$.
  \item For randomly and uniformly chosen $1\leq t\leq  T$, apply the
  next operation $W$ $t$ times to $\bfR_1$ if the content of
  $\bfR_0$ is ``1.''
  \begin{enumerate}
  \item To any $\ket{i}\ket{j}$, perform the next steps:
    (i) Check if $i\in V'$ with cost at most $C_F$,
    (ii) If $i\not\in V'$, apply diffusion operator $2C-I$ with cost
    at most $C_W$.
  \item To any $\ket{i}\ket{j}$, perform the next steps:
    (i) Check if $j\in V'$ with cost at most $C_F$,
    (ii) If $j\not\in V'$, apply diffusion operator $2R-I$ with cost at most $C_W$.
  \end{enumerate}
\item Apply the Hadamard operator to $\bfR_0$, and measure
  registers $\bfR_0$
  and $\bfR_1$ with respect to the computational basis.
\item If the result of measuring $\bfR_0$ is 1 or a marked
  element is found by measuring $\bfR_1$, output ``true''; otherwise output ``false.''
  \end{enumerate}
\end{theorem}

\section{Claw Detection}
\label{sec:claw-detection}
In this section, we describe ``claw-detection'' algorithms
that detect the existence of a claw.
The claw-detection algorithms will be used as subroutines
in the ``claw-search'' algorithms presented in
the next section that find a claw.

Before presenting the claw-detection algorithm, we introduce some notions.
The \emph{Johnson graph} $J(n,k)$ is a connected regular graph with 
${n \choose k}$ vertices such that every vertex is a subset of size
$k$ of $[n]$; two vertices are adjacent if and only if 
the symmetric difference of their corresponding subsets has size 2.
The \emph{graph categorical product} $G=(V_G,E_G)$ of two graphs $G_1=(V_{G_1},E_{G_1})$ and $G_2=(V_{G_2},E_{G_2})$, denoted
by $G=G_1\times G_2$, is a graph having vertex set $V_G=V_{G_1}\times V_{G_2}$
such that $((v_1,v_2),(v_1',v_2'))\in E_G$ if and only if
$(v_1,v_1')\in E_{G_1}$ and $(v_2,v_2')\in E_{G_2}$.

The next two propositions are useful in analyzing the claw-detection algorithms we
will describe. 
\begin{proposition}
\label{pr:spectral gap of graph product}
For Markov 
chains $\calM$, $\calM  _1, \dots ,\calM _k$,
the spectral gap $\delta $ of $\calM$ is the minimum of those
$\delta _1,\dots ,\delta _k$ of $\calM _1,\dots ,\calM _k$,
i.e., $\delta=\min_i\{\delta_i\}$,
  if the underlying graph of $\calM$ is the graph categorical product of
  those of $\calM _1,\dots ,\calM _k$.
\end{proposition}
The eigenvalues of the Markov chain on $J(n,k)$ are
$\frac{(k-j)(n-k-j)-j}{k(n-k)}$ for $j\in [0.k]$~\cite[pages
255--256]{brouwer-cohen-neumaier89}, from which the next proposition follows.
\begin{proposition}
\label{pr: spectral gap of Johnson graph}
The Markov chain on Johnson graph $J(n,k)$ has spectral gap
$\delta=\Omega(\frac{1}{k})$, if ${\;2\leq k\leq n/2}$.
\end{proposition}
We will first describe a claw-detection algorithm against a comparison oracle, 
from which we can almost trivially obtain a claw-detection algorithm against a standard oracle.
Let Claw\_Detect denote the algorithm.
To construct Claw\_Detect,
we apply Theorem~\ref{th:szegedy04}
on the graph categorical
product of two Johnson graphs $J_f=J(|X|,l)$ and $J_g=J(|Y|,m)$
for the domains $X$ and $Y$ of functions $f$ and $g$, respectively,
where $l$ and $m$ $(l\leq m)$ are integers fixed later.

More precisely, 
let $F$ and $G$ be any vertices of $J_f$ and $J_g$, respectively, i.e.,
any $l$-element subset and $m$-element subset of
$X$ and $Y$, respectively.
Then $(F,G)$ is a vertex in $J_f\times J_g$.
Similarly, for any edges $(F,F')$ and $(G,G')$ of $J_f$ and $J_g$, respectively,
$((F,G),(F',G'))$ is an edge 
connecting two vertices $(F,G)$ and $(F',G')$ in $J_f\times J_g$.
We next define ``marked vertices'' as follows.
Vertex $(F,G)$ is marked if there is a pair of
$(x,y)\in F\times G$ such that $f(x)=g(y)$.
To check if $(F,G)$ is marked or not,
we just sort all elements in $F\cup G$ on their function values.
Although we have to sort all elements in the initial vertex,
we have only to change a small part of the sorted list we have already had
when moving to an adjacent vertex.
For every vertex $(F,G)$, we maintain a representation $L_{F,G}$ of the sorted list
of all elements in $F\cup G$ on their function values, 
and we identify $(F,G, L_{F,G})$ as a vertex of $J_f\times J_g$. 
Here, we want to guarantee that
$L_{F,G}$ is uniquely determined for any pair 
$(F,G)$ in order to avoid
undesirable quantum interference;
we have just to introduce some appropriate rules that break ties, i.e.,
the situation where there are multiple elements in $F\cup G$ that have the same
function value.

As the state $\ket{\phi_0}$ in Theorem~\ref{th:szegedy04}, 
we prepare
$$
\ket{\phi_0}=\frac{1}{\sqrt{{N\choose l}{M\choose
      m}l(N-l)m(M-m)}}
\bigotimes _{\stackrel{|F\triangle F'|=|G\triangle
      G'|=2}{\stackrel{F,F'\subseteq X,
      |F|=|F'|=l}{\stackrel{G,G'\subseteq Y,|G|=|G'|=m}{}}}}
\ket{F,G,L_{F,G}}
\ket{F',G',L_{F',G'}},
$$
in register $\bfR_1$.
The number $1\leq t\leq \frac{c}{\sqrt{\delta \epsilon}}$ of repeating
  $W$ is chosen randomly and uniformly
for some constant $c$, $\delta := \Omega(1/m)$
and $\epsilon := lm/(NM)$.

We next describe the implementation of operation $W$.
Since diffusion operator $2C-I$ depends on $L_{F,G}$'s, it cannot be performed
without queries to the oracle. We thus divide operator $2C-I$ into
a few steps. 
For every unmarked vertex $(F,G,L_{F,G})$, we first transform 
$\ket{F,G,L_{F,G}}\ket{F',G',L_{F',G'}}$
into $\ket{F,G,L_{F,G}}\ket{F',G',L_{F,G}}$ with queries to the oracle.
We then perform a diffusion operator on the registers 
where the contents ``$F,G$'' and ``$F',G'$'' are stored,
to obtain a superposition of $\ket{F,G,L_{F,G}}\ket{F'',G'',L_{F,G}}$
over all $(F'',G'')$ adjacent to $(F,G)$.
Finally, we transform 
$\ket{F,G,L_{F,G}}\ket{F'',G'',L_{F,G}}$
into $\ket{F,G,L_{F,G}}\ket{F'',G'',L_{F'',G''}}$.
Operator $2R-I$ can be implemented in a similar way.

\begin{lemma}
\label{lm:claw-detect}
Let $Q_2({\bf claw}_{\rm detect}(N,M))$ be the number of queries needed
to decide whether there is a claw or not 
for functions
$f:X:=[N]\rightarrow Z$ and $g:Y:=[M]\rightarrow Z$ given as a
comparison oracle. Then,
\[
Q_2({\bf claw}_{\rm detect}(N,M))=
\begin{cases}
O((NM)^{1/3}\log N) & (N\leq M<N^2)\\
O(M^{1/2}\log N) & (M\geq N^2).
\end{cases}
\]
\end{lemma}
\begin{proof}
We will estimate $C_U$, $C_F$ and $C_W$ for Claw\_Detect,
and then apply Theorem~\ref{th:szegedy04}.

To generate $\ket{\phi _0}$, we first prepare the uniform
superposition of $\ket{F,G}\ket{F',G'}$ over all $F,F',G,G'$ such that
$(F,F')$ and $(G,G')$ are edges of $J_f$ and $J_g$, respectively.
Obviously, this requires no queries.  We then compute $L_{F,G}$ and
$L_{F',G'}$ for each basis state by issuing $O((l+m)\log (l+m))$
queries to oracle $O_{f,g}$. Thus, $C_U=O((l+m)\log (l+m))$.

We can check if there is a pair of 
$(x,y)\in F\times G$
such that $f(x)=g(y)$ by looking through $L_{F,G}$ (without any queries).
Thus, $C_F=0$.

For every unmarked $(F,G,L_{F,G})$,
step (a).ii of operation $W$ transforms
$\ket{F,G,L_{F,G}}\ket{F',G',L_{F',G'}}$ into a superposition
over all $\ket{F,G,L_{F,G}}\ket{F'',G'',L_{F'',G''}}$ such that
$|F\triangle F''|=|G\triangle G''|=2$.
This is realized by insertion and deletion of $O(1)$ elements
to/from the sorted list of $O(l+m)$ elements,
and diffusion operators without queries.
Each insertion or deletion can be performed 
with $O(\log (l+m))$ queries
by using binary search.
Similarly, step (b).ii of operation $W$ needs $O(\log (l+m))$ queries.
Thus, we have $C_W=O(\log (l+m))$.

We set $\epsilon$ to $\frac{l}{N}\times \frac{m}{M}$,
since the probability that a state is marked is minimized 
when only one claw exists for $f$ and $g$, in which case the
probability is $\frac{l}{N}\times \frac{m}{M}$.
Since, from Proposition~\ref{pr: spectral gap of Johnson graph}, 
the spectral gaps of the Markov chains on $J(N,l)$ and $J(M,m)$
are $\Omega(\frac{1}{l})$ and $\Omega(\frac{1}{m})$, respectively,
the spectral gap of the Markov chain on $J(N,l)\times J(M,m)$ 
is $\Omega(\min\{ \frac{1}{l}, \frac{1}{m}
\})=\Omega(\frac{1}{m})$ due to $l\leq m$ and
Proposition~\ref{pr:spectral gap of graph product}.

From Theorem~\ref{th:szegedy04}, the total number of queries is
$Q_2({\bf claw}_{\rm detect}(N,M))=O((l+m)\log(l+m)+\log(l+m)\sqrt{m(NM/(lm))})=O((l+m)\log(l+m)+\sqrt{NM/l}\log (l+m))$.

When $N\leq M<N^2$,
we set
$
l=m=\Theta((NM)^{1/3}),
$
which satisfies condition $l\leq N$.
The total number of queries is 
$
Q_2({\bf claw}_{\rm detect}(N,M))=O((NM)^{1/3}\log N).
$
When $M\geq N^2$, we set $l=m=N$,
implying that
$
Q_2({\bf claw}_{\rm detect}(N,M))=O(M^{1/2}\log N).
$
\qed
\end{proof}

The standard oracle case can be handled by using almost the same approach.
\begin{corollary}
Let $Q_2({\bf claw}_{\rm detect}(N,M))$ be the number of queries needed
to decide whether there is a claw or not 
for functions
$f:X=[N]\rightarrow Z$ and $g:Y=[M]\rightarrow Z$ given as a
standard oracle. Then,
\[
Q_2({\bf claw}_{\rm detect}(N,M))=
\begin{cases}
O((NM)^{1/3}) & (N\leq M<N^2)\\
O(M^{1/2}) & (M>N^2).\\
\end{cases}
\]
\label{cr:claw-detect}
\end{corollary}
The claw-detection algorithm against a standard oracle can easily be modified in order to solve the more general
problem
of \emph{detecting}
a tuple $(x_1,\dots ,x_p,y_1,\dots ,y_q)\in X^p\times Y^q$ 
such that
$x_i\neq x_j$ and $y_i\neq y_j$ for any $i\neq j$, and 
$(f(x_1),\dots ,f(x_p),g(y_1),\dots ,g(y_q))\in R$,
for given $R\subseteq Z^{p+q}$,
where $p$ and $q$ are any constant positive integers.
A modification is made to
the part of the algorithm that
decides whether a vertex of the underlying graph is marked or not; the
modification can be made without changing the number of queries.
The query complexity can be analyzed by using almost the same approach
as used in claw detection
with $\epsilon ={N-p \choose l-p}/{N \choose l}\times {M-q \choose m-q}/{M \choose
  m}\geq l^pm^q(1-o(1))/(N^pM^q) $;
the query complexity is
$O((N^pM^q)^{1/(p+q+1)})$ for $N\leq M<N^{1+1/q}$ and
$O(M^{q/(1+q)})$ for $M\geq N^{1+1/q}$.
The problem of \emph{finding} such a tuple can also be solved
with the same order of complexity as above
by using the algorithm for 
\emph{detecting} it as a subroutine.

Our algorithm for detecting a claw can easily be generalized to the
case of $k$ functions of domains of size $N_1, \dots ,N_k$,
respectively. More concretely, we apply Theorem~\ref{th:szegedy04} to the Markov chain
on the graph categorical product of the $k$ Johnson graphs, each of
which corresponds to one of the $k$ functions.
We denote this ``$k$-claw detection'' algorithm by
$k$-Claw\_Detect
in the next section.

\begin{lemma}
\label{lm:k-claw-detect}
For any positive integer $k>1$, 
let $Q_2({\bf k\mbox{-}claw}_{\rm detect}(N_1,\dots, N_k))$ 
be the number of queries
needed to decide whether there is a $k$-claw or not 
for functions
$f_i:X_i:=[N_i]\rightarrow Z\ (i\in [k])$ given as a
comparison oracle, where $N_i\leq N_j$ if $i<j$.
If $k$ is constant,
\[
Q_2({\bf k\mbox{-}claw}_{\rm detect}(N_1,\dots,N_k))=
\left\{
  \begin{array}{ll}
O\left(\left(\prod_{i=1}^k N_i\right)^\frac{1}{k+1}\log N_1\right) & \mbox{if $\prod_{i=2}^{k}N_i=O(N_1^k)$},\\
O\left(\sqrt{{\prod_{i=2}^k N_i}/{N_1^{k-2}}}\log N_1\right)
&\mbox{otherwise.}
  \end{array}
\right.
\]
\end{lemma}
\begin{proofsketch}
  In a way similar to the case of two functions,
we apply Theorem~\ref{th:szegedy04} on the graph categorical
product of $k$ Johnson graphs $J_{f_i}:=J(|X_i|,l_i)\ (i\in [k])$
for the domains $X_i$'s of functions $f_i$'s,
where $l_i$'s are integers fixed later such that $l_i\leq l_j$ for $i<j$.

To generate $\ket{\phi _0}$, we first prepare the uniform
superposition of $\ket{F_1,\dots, F_k}\ket{F_1',\dots,F_k'}$ over
all $F_i$ and $F'_i$ such that $(F_i,F'_i)$ is an edge of $J_{f_i}$
for every $i$.  This requires no queries.  
As in the case of two
functions,
define $L_{F_1,\dots,F_k}$ for any $F_1,\dots , F_k$ 
as a representation 
of the sorted list of all elements in $\bigcup _{i=1}^{k}F_i$
so that it can be uniquely
determined for each tuple $(F_1,\dots, F_k)$.
We then compute $L_{F_1,\dots, F_k}$
and $L_{F'_1,\dots, F'_k}$ for each basis state by issuing
$O\left((l_1+\dots +l_k)\log (l_1+\dots +l_k)\right)$ queries to the oracle.
Thus, $C_U=O\left((l_1+\dots +l_k)\log (l_1+\dots +l_k)\right)$.  $C_F$ and
$C_W$ can be estimated as 0 and $O\left(\log (l_1+\dots +l_k)\right)$,
respectively, in a way similar to the case of two functions.
We set $\epsilon$ to $\prod _{i=1}^{k}{l_i}/{N_i}$ and
$\delta$ to $\min_i \{ {1}/{l_i}\}={1}/{l_k}$.

When $\prod_{i=2}^kN_i=O(N_1^k)$,
we set
$l_i:=\Theta\left(\left(\prod_{i=1}^kN_i\right)^{\frac{1}{k+1}}\right)$
for every $i\in [k]$,
which satisfies condition $l_i\leq N_1\leq N_i$ for every $i\in [k]$.
When $\prod_{i=2}^kN_i=\Omega(N_1^k)$, we set 
$l_i:=\Theta(N_1)$ for every $i\in [k]$.
\end{proofsketch}
Against a standard oracle, we obtain a similar result.
\begin{corollary}
\label{cr:k-claw-detect}
For any positive integer $k>1$, let $Q_2({\bf k\mbox{-}claw}_{\rm detect}(N_1,\dots, N_k))$ 
be the number of queries needed
to decide whether there is a $k$-claw or not 
for functions
$f_i:X_i:=[N_i]\rightarrow Z$ $(i\in [k])$ given as a
standard oracle, where $N_i\leq N_j$ if $i<j$.
If $k$ is constant,
\[
Q_2({\bf k\mbox{-}claw}_{\rm detect}(N_1,\dots,
N_k))=
\left\{
  \begin{array}{ll}
O\left(\left(\prod_{i=1}^k N_i\right)^\frac{1}{k+1}\right) & \mbox{if $\prod_{i=2}^{k}N_i=O(N_1^k)$},\\
O\left(\sqrt{{\prod_{i=2}^k N_i}/{N_1^{k-2}}}\right)
&\mbox{otherwise.}
  \end{array}
\right.
\]
\end{corollary}

\section{Claw Finding}
\label{sec:claw-finding}
We now describe an algorithm, Claw\_Search, that finds a claw.
The algorithm consists of three stages.
In the first stage,
we find an $O(N)$-sized subset $Y'$ of $Y$
such that there is a claw in $X\times Y'$,
by performing binary search over $Y$ with Claw\_Detect.
In the second stage, we perform $4$-ary search 
over $X$ and $Y'$ with Claw\_Detect
to find $O(1)$-sized subsets $X''$ and $Y''$ of $X$ and $Y'$, respectively,
such that there is a claw in $X''\times Y''$.
In the final stage, we search $X''\times Y''$ for a claw 
by issuing classical queries.
To keep the error rate moderate, say, at most $1/3$, 
Claw\_Detect is repeated $O(s)$ times 
against the same pair of domains at the $s$th node of the
search tree at each stage.
This pushes up the query complexity by only a constant multiplicative
factor.

Figure~\ref{fig:claw search} precisely describes Claw\_Search.
Steps 2, 3 and 4 in the figure correspond to the first, second and
final stages, respectively.
\begin{figure}[htb]
\begin{center}
\hrulefill\\
  \textbf{Algorithm Claw\_Search}
\vspace{-2mm}
\begin{description}
  \item[Input:] Integers $M$ and $N$ such that $M\geq N$;
Comparison oracle $O_{f,g}$ for
    functions
$f:X\rightarrow Z$ and $g:Y\rightarrow Z$, respectively,  such
that $X:=[N]$ and $Y:=[M]$.\\
  \item[Output:] Claw pair $(x,y)\in X\times Y$ such that $f(x)=g(y)$ if
    such a pair exists; otherwise $(-1,-1)$.
\end{description}
\vspace{-5mm}
\begin{enumerate}
\item Set $\tilde{X}:=X$ and $\tilde{Y}:=Y$.
  \item  Set $s:=1$, and repeat the next steps until $u_{\tilde{Y}}-l_{\tilde{Y}}\leq |\tilde{X}|$, where $u_{\tilde{Y}}$ and $l_{\tilde{Y}}$ are the largest and
  smallest values, respectively, in $\tilde{Y}$.
  \begin{enumerate}
   \item Set $\Xi_Y:=\{ [l_{\tilde{Y}}.m_{\tilde{Y}}-1],[m_{\tilde{Y}}.u_{\tilde{Y}}]\}$, where $m_{\tilde{Y}}:=\lceil (l_{\tilde{Y}}+u_{\tilde{Y}})/2\rceil$.
   \item For every $\tilde{Y}'\in \Xi_Y$, do the following.\\
     If all $\tilde{Y}'\in \Xi_Y$ are examined, output $(-1,-1)$ and halt.
     \begin{enumerate}
     \item Apply Claw\_Detect $(s+2)$ times to
       $f$ and $g$ restricted to domains $\tilde{X}$ and
       $\tilde{Y}$, respectively. \\ 
     \item If at least one of the $(s+2)$ results is ``true,'' set $\tilde{Y}:=\tilde{Y}'$, and break (leave (b)).
     \end{enumerate}
   \item Set $s:=s+1$.
  \end{enumerate}
  \item Set $s:=1$, and repeat the next steps until $u_D-l_D\leq c$
  for every $D\in \{\tilde{X},\tilde{Y}\}$ and some constant $c$, say, 100, where $u_D$ and $l_D$ are the largest and
  smallest values, respectively, in $D$.
  \begin{enumerate}
   \item For every $D\in \{\tilde{X},\tilde{Y}\}$, set $\Xi_D:=\{ [l_D.u_D]\}$ if
     $u_D-l_D\leq c$, and\\ otherwise, set $\Xi_D:=\{ [l_D.m_D-1],[m_D.u_D]\}$ where $m_D:=\lceil (l_D+u_D)/2\rceil$.
   \item For every pair $(\tilde{X}',\tilde{Y}')\in
     \Xi_{\tilde{X}}\times \Xi_{\tilde{Y}}$, do the following.\\
     If all the pairs are examined, output $(-1,-1)$ and halt.
     \begin{enumerate}
     \item Apply Claw\_Detect $(s+3)$ times to
       $f$ and $g$ restricted to domains $\tilde{X}'$ and
       $\tilde{Y}'$, respectively. \\ 
     \item If at least one of the $(s+3)$ results is ``true,'' set $\tilde{X}:=\tilde{X}'$ and $\tilde{Y}:=\tilde{Y}'$, and break (leave (b)).
     \end{enumerate}
   \item Set $s:=s+1$.
  \end{enumerate}
\item Classically search $\tilde {X}\times \tilde {Y}$ for a claw.
\item Output claw $(x,y)\in \tilde {X}\times \tilde{Y}$ if it exists; otherwise output $(-1,-1)$.
\end{enumerate}
\vspace{-1.5\baselineskip}
\hrulefill
\end{center}
\vspace{-1.5\baselineskip}
\caption{Algorithm Claw\_Search}
\label{fig:claw search}
\vspace{-1\baselineskip}
\end{figure}
\begin{theorem}
\label{th:upper bound of claw}
Let $Q_2({\bf claw}_{\rm finding}(N,M))$ be the number of queries
needed to locate a claw if it exists
for functions
$f:X=[N]\rightarrow Z$ and $g:Y=[M]\rightarrow Z$ given as a
comparison oracle. Then,
\[
Q_2({\bf claw}_{\rm finding}(N,M))=
\begin{cases}
O\left((NM)^{1/3}\log N\right) & N\leq M <N^2\\
O(M^{1/2}\log N) & M\geq N^2.
\end{cases}
\]
\end{theorem}
\begin{proof}
We will analyze Claw\_Search in \fig~\ref{fig:claw search}.

When there is no claw, Claw\_Search always outputs the correct answer.
Suppose that there is a claw.
The algorithm may output a wrong answer 
if at least one of the following two cases happens.
In case (1), 
one of $O(\log {M/N})$ runs of step 2.(b) errs;
in case (2), 
one of $O(\log {N})$ runs of step 3.(b) errs.

Without loss of generality, the error probability of Claw\_Detect can
be assumed to be at most 1/3.
The error probability of each single run of step~2.(b).i is 
at most
$\frac{1}{3^{s+2}}$.
The error probability of each run of step~2.(b) is at most
$\frac{2}{3^{s+2}}<\frac{1}{3^{s+1}}$.
The error probability of case (1) is thus at most
$\sum _{s=1}^{\lceil \log M/N\rceil}\frac{1}{3^{s+1}}<\frac{1}{6}$.
The error probability of case (2) is also at most
$\sum _{s=1}^{\lceil \log
  N_1\rceil}\frac{1}{3^{s+1}}<\frac{1}{6}$ by similar calculation.
Therefore, the overall error probability is at most 1/6+1/6=1/3.

We next estimate the number of queries.
If $N\leq M <N^2$,
the size of $\tilde{Y}$ is always at most quadratically different from 
that of $\tilde{X}$.
Thus, the $s$th repetition of step~2 requires 
$O(s(NM/2^s)^{1/3}\log N)$ queries by Lemma~\ref{lm:claw-detect}.
Similarly, the $s$th repetition of step~3 requires 
$O(s(N/2^s)^{2/3}\log N)$ queries by Lemma~\ref{lm:claw-detect}.

The total number of queries is
$$
O
\left(
\sum _{s=1}^{\lceil\log (M/N)\rceil}
\left(
s\left(
N\frac{M}{2^s}
\right)^{1/3}
\log N
\right)
+
\sum _{s=1}^{\lceil\log N\rceil}
\left(
s(N/2^s)^{2/3}\log N
\right)
\right)
=O\left((NM)^{1/3}\log N\right).
$$

If $M \geq N^2$, 
the $s$th repetition of step~2 requires 
$O(s((NM/2^s)^{1/3}+ (M/2^s)^{1/2})\log N)$ by Lemma~\ref{lm:claw-detect}.
Thus, similar calculation gives $O(M^{1/2}\log N)$
queries.
\qed
\end{proof}
We can easily obtain the standard-oracle version of the above
theorem by using Corollary~\ref{cr:claw-detect} instead of Lemma~\ref{lm:claw-detect}.
\begin{corollary}
\label{cr:upper bound of claw}
Let $Q_2({\bf claw}_{\rm finding}(N,M))$ be the number of queries
needed to locate a claw if it exists
for functions
$f:X:=[N]\rightarrow Z$ and $g:Y:=[M]\rightarrow Z$ given as a
standard oracle. Then,
\[
Q_2({\bf claw}_{\rm finding}(N,M))=
\begin{cases}
O\left((NM)^{1/3}\right) & N\leq M <N^2\\
O(M^{1/2}) & M\geq N^2.
\end{cases}
\]
\end{corollary}
Similarly, we can find a $k$-claw by using $k$-Claw\_Detect as a subroutine.
First, we find $O(N_1)$-sized subset $X'_i$ of $X_i$ for every
$i\in [2.k]$
such that there is a $k$-claw in $X_1\times X'_2 \times \dots \times X'_k$,
by performing $2^{k-1}$-ary search over $X'_i$'s for all $i\in [2.k]$ with $k$-Claw\_Detect.
Let $X_1':=X_1$.
We then perform $2^{k}$-ary search 
over $X'_i$s for all $i\in [k]$ with $k$-Claw\_Detect
to find $O(1)$-sized subset $X''_i$ of $X'_i$ for every
$i\in [k]$
such that there is a $k$-claw in $X''_1\times \dots \times X''_k$.
Finally, we search $X''_1\times \dots \times X''_k$ for a $k$-claw
by issuing classical queries.
A more precise description of 
the algorithm, $k$-Claw\_Search, is given in \fig~\ref{fig:k-claw search}.
\begin{figure}[htb]
\begin{center}
\hrulefill\\
  \textbf{Algorithm $k$-Claw\_Search}
\vspace{-2mm}
\begin{description}
  \item[Input:] $k$ integers $N_1,\dots ,N_k$ such that $N_i\leq
    N_j$ if $i<j$.\\
Comparison oracle $O_{f_1,\dots, f_k}$ for
    functions
$f_i:X_i\rightarrow Z$ such
that $X_i:=[N_i]$  for every $i\in [k]$.\\
  \item[Output:] $k$-claw $(x_1, \dots, x_k)\in X_1\times \dots
    \times X_k$ such that $f_i(x_i)=f_j(x_j)$ for every $i,j\in [k]$ if
    it exists; otherwise $(-1,\dots ,-1)$.
\end{description}
\vspace{-5mm}
\begin{enumerate}
\item Set $\tilde{X_i}:=X_i$ for every $i\in [k]$.
  \item  Set $s:=1$, and repeat the next steps until $u_i-l_i\leq |\tilde{X}_1|$
  for all $i\in [2.k]$,
  where $u_i$ and $l_i$ are the largest and
  smallest values, respectively, in $\tilde{X}_i$.
  \begin{enumerate}
   \item For every $i\in [2.k]$, set $\Xi_i:=\{ [l_i.u_i]\}$ if
     $u_i-l_i\leq |\tilde{X}_1|$, and\\ otherwise, set $\Xi_i:=\{ [l_i.m_i-1],[m_i.u_i]\}$ where $m_i:=\lceil (l_i+u_i)/2\rceil$.
   \item For every tuple $(\tilde{X}'_1,\tilde{X}'_2,\dots , \tilde{X}'_k)\in
     \{\tilde{X_1}\}\times \Xi_2\times \dots \times \Xi_k$, do the following.\\
     If all the tuples are examined, output $(-1,\dots ,-1)$ and halt.
     \begin{enumerate}
     \item Apply $k$-Claw\_Detect $(s+1)+\ceil{\log _3 2^{k-1}}$ times to
       the $k$ functions $f_i$ restricted to domains $\tilde{X}'_i$, respectively, for
       every $i\in [k]$. 
     \item If at least one of the $(s+1)+\ceil{\log _3 2^{k-1}}$
       results is ``true,'' set $\tilde{X}_i:=\tilde{X}'_i$ for every $i\in
      [2.k]$, and break (leave (b)).
     \end{enumerate}
   \item Set $s:=s+1$.
  \end{enumerate}
  \item Set $s:=1$, and repeat the next steps until $u_i-l_i\leq c$
  for all $i\in [k]$ and some constant $c$, say, 100, where $u_i$ and $l_i$ are the largest and
  smallest values, respectively, in $\tilde{X}_i$.
  \begin{enumerate}
   \item For every $i\in [k]$, set $\Xi_i:=\{ [l_i.u_i]\}$ if
     $u_i-l_i\leq c$, and\\ otherwise, set $\Xi_i:=\{ [l_i.m_i-1],[m_i.u_i]\}$ where $m_i=\lceil (l_i+u_i)/2\rceil$.
   \item For every tuple $(\tilde{X}'_1,\tilde{X}'_2,\dots , \tilde{X}'_k)\in
     \Xi_1\times \dots \times \Xi_k$, do the following.\\
     If all the tuples are examined, output $(-1,\dots ,-1)$ and halt.
     \begin{enumerate}
     \item Apply $k$-Claw\_Detect $(s+1)+\ceil{\log _3 2^{k}}$ times to
       the $k$ functions $f_i$ restricted to domains $\tilde{X}'_i$ for
       every $i\in [k]$. 
     \item If at least one of the $(s+1)+\ceil{\log _3 2^{k}}$ results
       is ``true,'' set $\tilde{X}_i:=\tilde{X}'_i$ for every $i\in
      [k]$, and break (leave (b)).
     \end{enumerate}
   \item Set $s:=s+1$.
  \end{enumerate}
\item Classically search $\tilde {X}_1\times \dots \times \tilde{X}_k$ for a $k$-claw.
\item Output $k$-claw $(x_1,\dots ,x_k)\in X'_1\times \dots
  \times X'_k$ if it exists; otherwise output $(-1,\dots ,-1)$.
\end{enumerate}
\vspace{-6mm}
\hrulefill
\end{center}
\vspace{-1\baselineskip}
\caption{Algorithm $k$-Claw\_Search}
\label{fig:k-claw search}
\end{figure}

\begin{theorem}
\label{upper bound of k-claw finding}
For any positive integer $k>1$, let $Q_2({\bf k\mbox{-}claw}_{\rm finding}(N_1,\dots ,N_k))$ 
be the number of queries
needed to locate a $k$-claw if it exists
for $k$ functions
$f_i:X_i:=[N_i]\rightarrow Z\ (i\in [k])$ given as a
comparison oracle, where $N_i\leq N_j$ if $i<j$.
If $k$ is constant,
\[
Q_2({\bf k\mbox{-}claw}_{\rm finding}(N_1,\cdots ,N_k))
=
  \begin{cases}
O\left(\left(\prod_{i=1}^k N_i\right)^\frac{1}{k+1}\log N_1\right) & \mbox{if $\prod_{i=2}^{k}N_i=O(N_1^k)$},\\
O\left(\sqrt{{\prod_{i=2}^k N_i}/{N_1^{k-2}}}\log N_1\right)
&\mbox{otherwise.}
  \end{cases}
\]
\end{theorem}
We can easily obtain the standard-oracle version of the above
theorem by using Corollary~\ref{cr:k-claw-detect} instead of Lemma~\ref{lm:k-claw-detect}.
\begin{corollary}
For any positive integer $k>1$, let $Q_2({\bf k\mbox{-}claw}_{\rm finding}(N_1,\dots, N_k))$ 
be the number of queries needed
to locate a $k$-claw if it exists
for $k$ functions
$f_i:X_i:=[N_i]\rightarrow Z\ (i\in[k])$ given as a
standard oracle, where $N_i\leq N_j$ if $i<j$.
If $k$ is constant,
\[
Q_2({\bf k\mbox{-}claw}_{\rm finding}(N_1,\cdots,
N_k))=
  \begin{cases}
O\left(\left(\prod_{i=1}^k N_i\right)^\frac{1}{k+1}\right) & \mbox{if $\prod_{i=2}^{k}N_i=O(N_1^k)$},\\
O\left(\sqrt{{\prod_{i=2}^k N_i}/{N_1^{k-2}}}\right)
&\mbox{otherwise.}
  \end{cases}
\]
\end{corollary}

\end{document}